\documentclass[twoside, a4paper, 12pt]{amsart}
\usepackage{graphicx}
\usepackage[latin1]{inputenc}
\usepackage{amssymb,amsmath,amsaddr}
\usepackage{eucal}
\usepackage{mathrsfs}
\usepackage{xspace}
\usepackage[round,longnamesfirst]{natbib}
\usepackage{url}
\usepackage{color}
\usepackage[breaklinks=true]{hyperref}
\usepackage{breakcites}
\newtheorem{theorem}{Theorem}[section]
\newtheorem{proposition}[theorem]{Proposition}

\newtheorem{corollary}[theorem]{Corollary}

\theoremstyle{definition}
\newtheorem{definition}[theorem]{Definition}

\theoremstyle{remark}
\newtheorem{remark}[theorem]{Remark}
\newcommand*{\remaend}{\hfill\text{$\diamond$}}

\numberwithin{equation}{section}
\newcommand*{\N}{\mathbb{N}}
\newcommand*{\R}{\mathbb{R}}

\renewcommand*{\epsilon}{\varepsilon}
\newcommand*{\defined}{\triangleq}

\newcommand*{\E}{\mathbb{E}}
\newcommand*{\Cov}{\mathop{\rm Cov}}
\newcommand*{\Cor}{\mathop{\rm Cor}}
\newcommand*{\Var}{\mathop{\rm Var}}

\newcommand*{\filF}{\mathscr{F}}
\newcommand*{\Jc}{\mathbb{I}}

\newcommand{\iidsim}{\mathop{\stackrel{\rm\scriptscriptstyle i.i.d.}{\sim}}}
\urldef{\urlVW}\url{http://www.mathematik.uni-kl.de/~wwwfm/RobustRiskEstimation}

\newcommand*{\data}{{\rm data}}
\newcommand*{\simu}{{\rm sim}}
\newcommand*{\redeem}{{\rm redeem}}
\begin{document}

%-----------------------------------------------------------------------------------
\title[GPD Processes and Liquidity]{Generalized Pareto Processes\linebreak and Liquidity}

\author{Sascha~Desmettre$^{\dagger}$ $\cdot$ Johan~de~Kock$^{*}$ $\cdot$ Peter Ruckdeschel$^{\ddagger}$ $\cdot$ Frank~Thomas~Seifried$^{\S}$}

\address{$^\dagger$ \small Institute of Stochastics, Karlsruhe Institute of Technology (KIT)\linebreak Englerstra\ss{}e 2, 76131 Karlsruhe, Germany; sascha.desmettre@kit.edu;\\
Department of Mathematics, University of Kaiserslautern\linebreak Erwin-Schr\"odinger Stra\ss{}e, 67663 Kaiserslautern, Germany}
\address{$^*$ \small Liberty Life: Libfin Markets, Claremont Central \linebreak
Claremont, Cape Town, 7735, South Africa; johan.dekock@liberty.co.za}
\address{$^\ddagger$ \small Institute for Mathematics, University of Oldenburg, P.O.\ Box 2503, 26111 Oldenburg, Germany; peter.ruckdeschel@uni-oldenburg.de}
\address{$^\S$ \small Department IV -- Mathematics, University of Trier\linebreak Universit\"atsring 19, 54296 Trier, Germany; seifried@uni-trier.de}

\date{\today; \textit{Corresponding author:} Peter Ruckdeschel}

\begin{abstract} Motivated by the modeling of liquidity risk in fund management in a dynamic setting, we propose and investigate a class of time series models with generalized Pareto marginals:
the autoregressive generalized Pareto process (ARGP), a modified ARGP (MARGP) and a thresholded ARGP (TARGP).
These models are able to capture key data features apparent in fund liquidity data and reflect the underlying phenomena via easily interpreted, low-dimensional model parameters.
We establish stationarity and ergodicity, provide a link to the class of shot-noise processes, and determine the associated interarrival distributions for exceedances.
Moreover, we provide estimators for all relevant model parameters and establish consistency and asymptotic normality for all estimators (except the threshold parameter, which as usual must be dealt with separately).
Finally, we illustrate our approach using real-world fund redemption data, and we discuss the goodness-of-fit of the estimated models.\\


\end{abstract}

\maketitle
%-----------------------------------------------------------------------------------

%-----------------------------------------------------------------------------------
\section{Introduction}\label{SectIntroduction}
%-----------------------------------------------------------------------------------

%-----------------------------------------------------------------------------------
\paragraph{\textit{Motivation: Liquidity Risk in Fund Management.}}
In general, liquidity risk refers to the risk that cash or other liquid means of payment are not available when they are required, or only at increased cost.
In the context of fund management, a particular focus is on calling risk, i.e.\ the risk of unplanned withdrawals such as, for instance, early redemptions of shares in a mutual fund.
In that context, \cite{Fiedler2000} established the notions of expected and dynamic liquidity-at-risk (ELaR and DyLaR) as similar concepts to the value-at-risk (VaR), modeling the extremes of the expected cash liquidity as quantiles of the underlying distributions.

Over the past few years, professional management of redemption risks has become an important regulatory requirement, both in the European Union and in the United States; see the corresponding guidelines \cite{UCITS}, \cite{AIFM}, \cite{SEC}.
In \cite{DesmettreDeege2016} the authors propose a static model to quantify these redemption risks specific to mutual funds via the {\it liquidity at risk (LaR)}, which they adapted from the measurement of daily net cahsflows specific to the banking sector (compare \cite{Zeranski06}), and which complies with regulatory guidelines.
The liquidity at risk is based on the well-known peaks-over-threshold approach from extreme value theory (see, for instance, \cite{Embrechts}), providing a liquidity reserve that is not exceeded with a certain probability.

This paper is concerned with a {\it dynamic} model for fund liquidity risks that retains the key distributional features of the static approach, including GPD marginals.
%-----------------------------------------------------------------------------------

%-----------------------------------------------------------------------------------
\paragraph{\textit{Feature-Based Modeling and ARGP Processes.}}
One natural theoretical approach for a dynamic GPD model is to provide a functional (or process) version of the Pickands-Balkema-de-Haan theorem (see, e.g., \cite{Embrechts}) to construct a discrete- or continuous-time GPD process as a limiting object.
This route has been taken successfully by \cite{FerreiraHaan2012} and \cite{DombryRibatet2015}.
In view of the applications we aim for, the analysis of this article takes a more data-driven perspective, based on the {\it feature-based modeling approach} of P.~L.~Davies \citep{Davies1995}.
This approach has successfully been pursued in many different domains, compare, e.g., \cite{Dumb:98,WWLK:05,Li:Liu09,Hen:10}, and can be summarized as follows:
First, one identifies reproducible, pertinent features of observed realizations.
Second, one designs parsimonious statistical models that are able to reproduce these features in a natural way.

Concerning potential parametrizations of our models, we follow Occam's razor (or the parsimony principle of statistics) and head for statistical models with low-dimensional and easily interpreted parametrizations.
We are thus lead to consider a Markov process, the autoregressive generalized Pareto (ARGP) process, which naturally generalizes the Pareto processes introduced by \cite{YehArnoldRobertson1988}.
To add flexibility in the modeling of increasing sequences of exceedances, we also introduce the modified ARGP (MARGP) process.
To be able to model censored data, as observed in the context of, e.g., fund redemptions, we further define a thresholded variant of the (M)ARGP process (TARGP).
%-----------------------------------------------------------------------------------

%-----------------------------------------------------------------------------------
\paragraph{\textit{Related Models.}}
To the best of our knowledge, there are no canonical dynamic models for funds' liquidity risks, or for dynamic GPD processes.
The ARGP models we propose in this article are inspired by, and can be seen as a natural generalization of, the Pareto processes introduced by \cite{YehArnoldRobertson1988}; see Corollary~\ref{Corof}~(b) below.
ARGP processes also share some structural features with shot-noise processes that run backwards in time.
Shot-noise processes have been introduced in a financial risk context by \cite{KlueppelbergMikosch95} and further generalized by \cite{SchmidtStute2007} and \cite{AltmannSchmidtStute2008}; applications to, e.g., credit risk, can be found in \cite{GasparSchmidt2005,GasparSchmidt2007,GasparSchmidt2010} and \cite{SchererSchmidSchmidt2012}.
In contrast to these continuous-time models, the ARGP class is formulated in discrete time, and it guarantees GPD marginals.
%-----------------------------------------------------------------------------------

%-----------------------------------------------------------------------------------
\paragraph{\textit{Organization of the Article.}}
In Section~\ref{DataFeatures} we identify the relevant data features in realized fund redemptions.
Based on this, in Section~\ref{SectARGPProcesses} we introduce ARGP processes as natural statistical models for these time series.
Thus we define the ARGP process in Section~\ref{SectARGPProcess}, the MARGP process in Section~\ref{carmouflage}, and the TARGP process in Section~\ref{SectTruncARGP}.
Moreover we establish the key theoretical properties of these processes, including stationarity, marginal distributions, and ergodicity, and we clarify the link to shot-noise processes.
Section~\ref{InterARR} determines the interarrival distributions for the ARGP process and its variants.
In Section~\ref{ParamEst} we provide estimators for the model parameters of the ARGP processes, and we establish consistency and joint asymptotic normality of all proposed estimators (except the GPD threshold parameter, which as usual is to be dealt with separately).
In Section~\ref{SectResults} we illustrate our model and methodology in a benchmark setting using real-world fund redemption data, and Section~\ref{SectConclusion} concludes.

\newpage
%-----------------------------------------------------------------------------------

%-----------------------------------------------------------------------------------
\section{Data Features of Fund Redemption Time Series}\label{DataFeatures}
%-----------------------------------------------------------------------------------

%-----------------------------------------------------------------------------------
\citet{Davies1995} introduces \textit{data features} as a fundamental notion for statistical modeling.
Formally, data features can be regarded as functionals from the underlying distributions to some feature space.
Since in applications one wishes to infer these features from historical data, the domains of these data feature functionals should contain the observed empirical distributions and, with a view towards inference, should ensure that the empirical distribution is close to the model distribution, i.e.\ the data feature functional should be weakly continuous.
As a basis for the models proposed in Section~\ref{SectARGPProcesses}, in this section we identify several key data features in typical time series of fund redemptions.
We refrain from formally specifying the feature space and checking weak continuity, as it should be clear that this is feasible.
%-----------------------------------------------------------------------------------

%-----------------------------------------------------------------------------------
\paragraph{\textit{Fund Redemption Cashflows.}}
We consider a typical time series of cashflow redemptions in a popular open end mutual fund as displayed in Figure~\ref{fig:Example}.

\begin{figure}[hb]
\includegraphics[width=\textwidth]{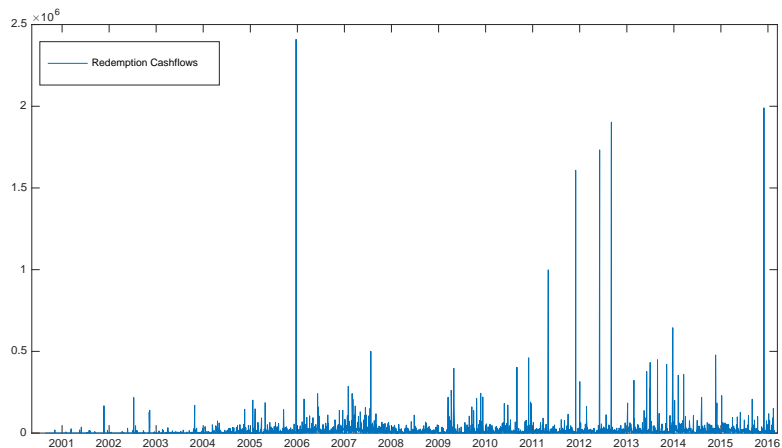}
\caption{Typical redemption cashflow time series for a mutual fund.}\label{fig:Example}
\end{figure}

We note that the data is characterized by different regimes of fund redemptions; regimes of low or moderate redemption cashflows are followed by regimes of extraordinarily high redemptions.
In the second regime, the data in fact exceeds very high thresholds, and it appears that the distribution of exceedances is invariant when the threshold is increased.
We take this invariance property as the first data feature, and we subsequently refer to it as \textsf{DF1}.

Second, possibly excluding the initial time period, the process of redemptions is apparently stationary in time.
We therefore adopt stationarity as the second data feature and refer to it as \textsf{DF2}.

Third, upon investigating exceedances more closely, it becomes apparent that we can distinguish two subregimes (see Figure~\ref{fig:DF}), and that otherwise exceedances look strikingly similar.
The fact that the process appears to switch randomly between these two subregimes is taken as data feature \textsf{DF3}.
As exemplified by Figure~\ref{fig:DF}, the subregimes are characterized either by (a) a deterministic build-up over time (data feature \textsf{DF4}), followed by a sudden and sharp decrease (data feature \textsf{DF5}); or (b) an isolated, single shock (data feature \textsf{DF6}).
The behavior in subregime~(a) is reminiscent of a shot-noise process run backward in time; we return to this link in Section~\ref{SectARGPProcess} below.

\begin{figure}[ht]
\includegraphics[width=0.495\textwidth,height=4cm]{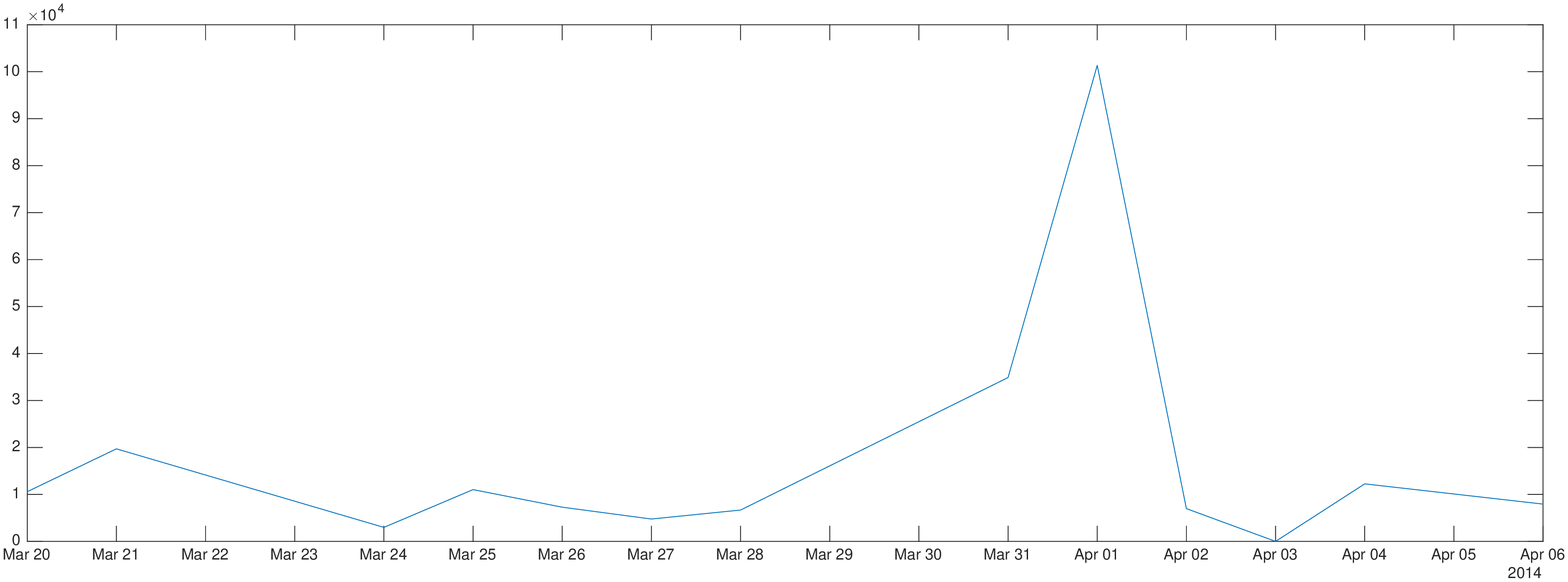}
\includegraphics[width=0.495\textwidth,height=4cm]{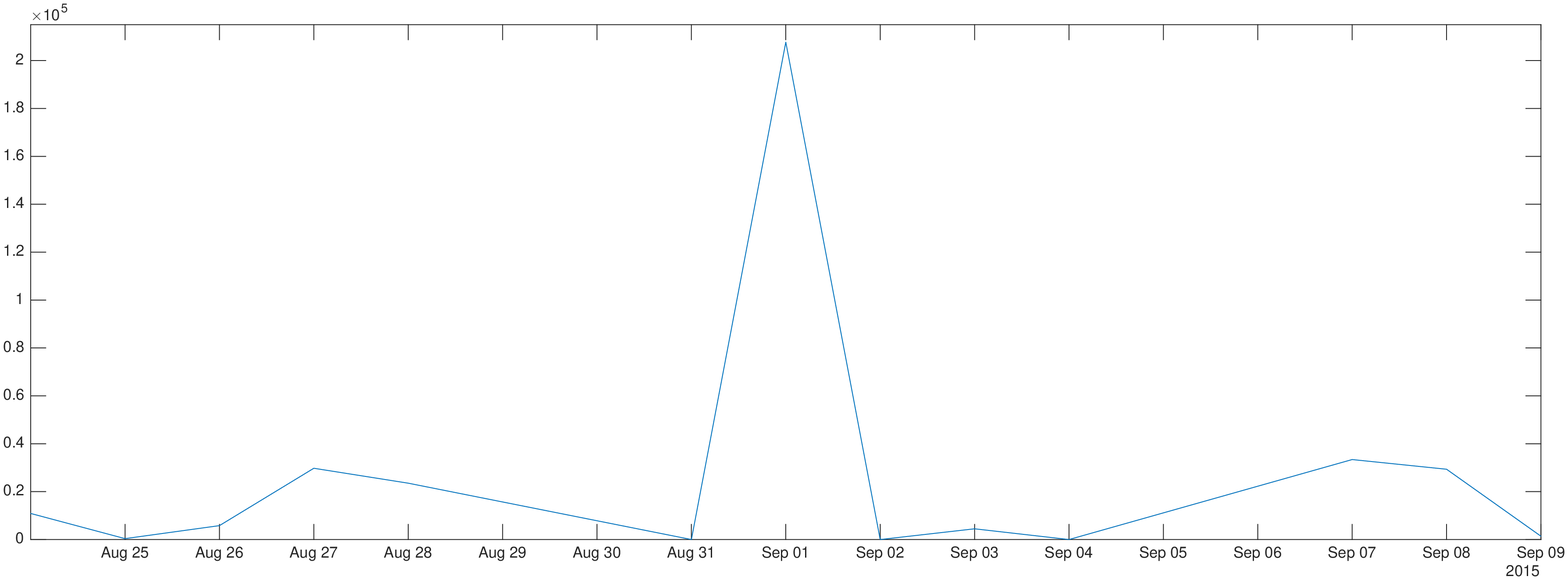}
\caption{Illustration of data features \textsf{DF4} and \textsf{DF5} (left panel) and \textsf{DF6} (right panel).}\label{fig:DF}
\end{figure}

An additional feature of the fund redemptions process is the fact that it can attain zero, reflecting business days without redemption cashflows.
We note this as data feature \textsf{DF7}.
Note that, when applied to fund liquidity risk, the ARGP process as defined in Definition~\ref{Def:ARGP} below covers only the process of exceedances \textit{after} thresholding at a threshold level $u\in\R$.
In a dynamic setting as considered in this paper, the interarrival times of such exceedances are relevant as well.
Hence we do not wish to discard the information that a non-exceedance has been observed at time $t$, but rather observe a censored pair $(\Jc(X>u),(X-u)_+)$, where for inference on the GPD only the non-censored observations, i.e.\ the instances $\Jc(X>0)=1$ are relevant, while for the discussion of interarrival times the runs of consecutive censorings $\Jc(X>0)=0$ carry information as well.
To combine both coordinates in a one-dimensional statistic, we can without loss restrict attention to the values $(X-u)_+$.
%-----------------------------------------------------------------------------------

%-----------------------------------------------------------------------------------
This concludes our discussion of pertinent data features.
In the next section, we propose statistical models that are able to match \textsf{DF1}-\textsf{DF7} both qualitatively and quantitatively.
Clearly, these are not the only reasonable models to capture the above data features; the aim of this article is to construct and analyze a parsimonious class of models that are suitable for this purpose.
Moreover, although \textsf{DF1}-\textsf{DF7} have been derived from, and feature prominently in, time series of fund redemptions, the following mathematical and statistical analysis applies whenever time series data share these basic features.

\newpage
%-----------------------------------------------------------------------------------

%-----------------------------------------------------------------------------------
\section{Processes Matching the Data Features}\label{SectARGPProcesses}
%-----------------------------------------------------------------------------------

%-----------------------------------------------------------------------------------
\subsection{The ARGP Process}\label{SectARGPProcess}
%-----------------------------------------------------------------------------------

%-----------------------------------------------------------------------------------
\paragraph{\textit{GPD Marginals.}}
In view of the Pickands-Balkema-de-Haan theorem, data feature \textsf{DF1} strongly suggests modeling the marginals using GPD distributions.
We briefly recall the definition:
\begin{definition}[GPD]\label{Def:GPD}
The cumulative distribution function (cdf) of the generalized Pareto distribution (GPD) with shape parameter $\xi\in\mathbb{R}$ and scale parameter $\sigma>0$ is given by
\[
G_{\xi,\sigma}(x) \defined
\begin{cases}
\displaystyle
  1 - \left(1 + \frac{\xi\,x}{\sigma}\right)^{-1/\xi}\,&\mbox{if}\quad\xi\neq0\,,\\
\displaystyle
 1-e^{-x/\sigma}\,&\mbox{if}\quad\xi=0\,,
\end{cases}
\quad \text{for }x\in D(\xi,\sigma)\,,
\]
where
\[
D(\xi,\sigma) \defined
\begin{cases}
\displaystyle
  [0,\infty) \,&\mbox{if}\quad\xi\geq0\,,\\
\displaystyle
  [0,-\sigma/\xi]\,&\mbox{if}\quad\xi <0\,.
\end{cases}
\]
\end{definition}
Concerning the ranges of the underlying model parameters, in view of the applications we wish to consider, there is no loss in assuming that
\[
\xi > -1/2\,.
\]
Under this condition, we can give a coherent treatment of, e.g., statistical inference, without having to consider particular cases.

In some references in the literature, an unknown threshold $u\in\R$ is included as an additional parameter, so the observations $X_i$ are of the form $X_i=X_i^{(0)}+u$ where $X_i^{(0)} \sim G_{\xi,\sigma}$.
We return to this discussion in the context of the truncated ARGP model in Section~\ref{SectTruncARGP}; until then, we assume that $u=0$.
%-----------------------------------------------------------------------------------

%-----------------------------------------------------------------------------------
\paragraph{\textit{Dynamics and Definition of the ARGP Process.}}
Returning to the data features in Section~\ref{DataFeatures}, we now address the dynamic features of exceedances identified in the data.

First, \textsf{DF3} can be accommodated by a binomial switching mechanism that, at each point of time, randomly selects between the two regimes.
The monotone build-up in subregime~(a) reflected in \textsf{DF4} can be modeled in terms of a deterministic, strictly increasing function $f$ that acts on the previous observation $X_{t-1}$.
On the other hand, the purely random exceedances \textsf{DF6} in subregime~(b) can be captured by an independent random variable $\epsilon_t$.
Finally, the harsh decrease following a sequence of subsequent increases in subregime~(a), codified as \textsf{DF5}, can be modeled by a second, competing risk mechanism that compares the current level of the process with another stochastically independent variable $\bar\epsilon_t$ and specifies the new value of the process as the minimum of the theoretical value of the process and $\bar\epsilon_t$.

Note that, since $\epsilon_t$ and $\bar\epsilon_t$ are only relevant in subregimes~(b) and (a), respectively, and since they are assumed to be independent, we can condense them into a single sequence of independent variables $\epsilon_t$.
Thus we arrive at the following Markovian recursion capturing data features \textsf{DF3}--\textsf{DF6}:
\begin{equation}\label{eq:ARGP}
X_t\defined U_t f(X_{t-1}) + (1-U_t)\min\left\{f(X_{t-1}), \epsilon_t\right\}
\end{equation}
where $\epsilon_t, X_{t-1}\sim G_{\xi,\sigma}$ and $U_t\sim{\rm Bernoulli}(\beta)$, and where $\epsilon_t, X_{t-1},U_t$ are independent.
The increasing function $f$ remains to be specified.

We now address the appropriate specification of $f$ so as to maintain stationarity. Denote the uniform distribution on $[0,1]$ by ${\rm unif}[0,1]$.
Since $G=G_{\xi,\sigma}$ is continuous, we can without loss of generality transform $X_t$ and $\epsilon_t$ to variables $X^\ast_t\defined G(X_t)$ and $\epsilon_t^\ast\defined G(\epsilon_t)$ 
distributed according to ${\rm unif}[0,1]$ and obtain
\begin{equation}\label{eq:unifX}
X^\ast_t = U_t f^\ast(X^\ast_{t-1}) + (1-U_t)\min\left\{f^\ast(X^\ast_{t-1}), \epsilon^\ast_t\right\}
\end{equation}
where $f^\ast\defined F\circ f\circ F^{-1}$, and we use the fact that $G(\min\{x,y\})=\min\{G(x),G(y)\}$.
Using \eqref{eq:unifX} we can show that $\{X_t\}$ has $G$ marginals;
more generally, the following result shows that a variant of the above Markovian recursion leads to a stationary process with any given marginal distribution.
\begin{proposition}\label{fdefProp}
Let $F$ be a continuous cdf, and let $F^{-1}$ denote its quantile function.
Suppose $X_0\sim F$, $\epsilon_t\iidsim F$ are independent, and
\[
f\defined F^{-1}\circ f^\ast \circ F\quad\text{where}\quad f^\ast(u)\defined\frac{u}{(1-\beta)u+\beta}\, ,\quad u\in[0,1]\,.
\]
Then the process defined in equation~\eqref{eq:ARGP} has $F$ marginals.
Moreover, $f$ is uniquely determined by this condition on $\{x\in\R\, :\, 0<F(x)<1\}$.
\end{proposition}
\begin{proof}
$\{X_t\}$ has $F$ marginals if and only if $\{X^\ast_t\}$ has ${\rm unif}[0,1]$ marginals; similarly, $\epsilon_t\sim F$ if and only if $\epsilon^\ast_t\sim{\rm unif}[0,1]$.
Suppose by induction that $X^\ast_{t-1}\sim{\rm unif}[0,1]$, and note that by independence
\begin{align*}
P(X_t^\ast > u) &= P(U_t=1,\,f^\ast(X^\ast_{t-1})>u)\\
&\qquad + P(U_t=0,\,f^\ast(X^\ast_{t-1})>u,\,\epsilon^\ast_t>u)\\
&= \beta P(f^\ast(X^\ast_{t-1}) > u) + (1-\beta)P(f^\ast(X^\ast_{t-1}) > u)(1-u)\\
&= P(f^\ast(X^\ast_{t-1}) > u) [1-(1-\beta)u] = 1-u
\end{align*}
where we use the fact that
\begin{equation}\label{Equafstarinverse}
P(f^\ast(X^\ast_{t-1}) > u) = 1-(f^\ast)^{-1}(u) = \frac{1-u}{1-(1-\beta)u}\,,\quad u\in[0,1].
\end{equation}
It follows that $X^\ast_t\sim{\rm unif}[0,1]$ and $X_t\sim F$, as asserted.
It is clear from the above that $(f^\ast)^{-1}$ is uniquely determined on $(0,1)$, hence so is $f^\ast$; thus $f$ is uniquely determined on $\{0<F<1\}$.
\end{proof}

Note that equation~\eqref{eq:ARGP} provides a method to construct processes with
arbitrary marginals specified by an increasing, continuous cdf including, e.g., Weibull, GEV or Gamma marginals.
In the case of generalized Pareto and Pareto marginals, we explicitly determine the relevant function $f$ in the following result.
\begin{corollary}\label{Corof}
\begin{enumerate}
\item[(a)] For the generalized Pareto cdf $G_{\xi,\sigma}$ we have
\begin{equation}\label{fdef}
f(x) = \tfrac{\sigma}{\xi} \left\{ \left( 1 + \tfrac{1}{\beta}\left[ (1+\tfrac{\xi x}{\sigma})^{\frac{1}{\xi}}-1\right] \right)^\xi - 1 \right\}\quad \text{for }x\in D(\xi,\sigma).
\end{equation}
\item[(b)] For the Pareto cdf $\tilde G_{\xi,\sigma}(x)=1-(1+(x/\sigma)^{1/\xi})^{-1}$ we obtain
\[
f(x) = \beta^{-\xi} x,\quad x\in[0,\infty).
\]
\end{enumerate}
\end{corollary}
Note that in the situation of part (b), the construction \eqref{eq:ARGP} recovers the classical ARP(1) process of \cite{YehArnoldRobertson1988}.
In the following, we focus on the GPD distribution and assume that
\[
f: D(\xi,\sigma)\to\R\quad \text{is defined by~\eqref{fdef}.}
\]
Note that $D(\xi,\sigma)$ is the closure of $\{0<G_{\xi,\sigma}<1\}$.
With this specification, the preceding construction gives rise to the ARGP process:\footnote{
Note that Definition~\ref{Def:ARGP} does not require that $X_0\sim G_{\xi,\sigma}$.}
\begin{definition}[ARGP Process]\label{Def:ARGP}
Suppose that $U_t \iidsim {\rm Bernoulli}(\beta)$ and $X_0$ and $\epsilon_t \iidsim G_{\xi,\sigma}$ are independent, and define $X_t$ recursively by
\[
X_t\defined U_t f(X_{t-1}) + (1-U_t)\min\left\{f(X_{t-1}), \epsilon_t\right\}\quad\text{for }t\in\N.
\]
Then the process $\{X_t\}_{t\in\N_0}$ is called an \textit{autoregressive generalized Pareto} (\textit{ARGP}) process with parameters $\xi,\sigma,\beta$.
\end{definition}

Note that by construction, the frequencies of regimes~(a) and (b) are controlled by the binomial probability $\beta$.
High values of $\beta$ produce frequent build-up phases, i.e.\ regime~(a), while low values of $\beta$ are more likely to lead to singular exceedances, i.e.\ regime~(b).
%-----------------------------------------------------------------------------------

%-----------------------------------------------------------------------------------
\paragraph{\textit{Stationarity.}}
We now address \textsf{DF2}, i.e.\ stationarity of the ARGP process.
In fact, we can show that the ARGP process is ergodic:
\begin{proposition}\label{alphamix}
If $X_0\sim G_{\xi,\sigma}$ then $\{X_t\}$ is stationary.
For $\beta<1$ and any distribution of $X_0$, the process $\{X_t\}$ is $\alpha$-mixing, hence asymptotically stationary and ergodic.
\end{proposition}
\begin{proof}
Since $X_0\sim G_{\xi,\sigma}$, Proposition~\ref{fdefProp} implies that all marginals are $G_{\xi,\sigma}$.
We have $\{X_0,\epsilon_s\}\iidsim G_{\xi,\sigma}$, and this family is also independent from $\{U_t\}\iidsim {\rm Bernoulli}(\beta)$.
Moreover, we have $X_t=g(U_t,X_{t-1},\epsilon_t)$ with the deterministic function
\[
g(u,x,e) \defined (1-u)f(x)+u\min\{f(x),e\}.
\]
Thus Propositions~6.4 and 6.31 in \cite{Breiman92} apply, and it follows that $\{X_t\}$ is stationary and ergodic.
To show that $\{X_t\}$ is $\alpha$-mixing, let $\filF_0^m\defined \sigma(X_0,\ldots,X_m)$, $\filF_m^\infty\defined\sigma(X_m,X_{m+1},\ldots)$ and
\[
\alpha_h \defined \sup \{|P(A\cap B)-P(A)P(B)| \,:\, A\in \filF_0^n,\ B\in\filF_{n+h}^\infty,\ n\in\N\}.
\]
Then for each $\delta>0$ there is some $N\in\N$ such that
\[
P(\epsilon^\ast_{t+n}>(f^\ast)^n(X^\ast_t))\leq \delta\quad \text{for all }n>N
\]
where $f^n=f\circ\,\cdots\,\circ f$ ($n$ times).
Hence for some constant $C=C(N,f^\ast)$
\[
\alpha_h\leq P(X_{t+h}=f^h(X_t)) \leq C (\beta+\delta)^h\quad \text{for all }h\in\N.
\]
Thus $\{X_t\}$ is $\alpha$-mixing (or strongly mixing), and in particular asymptotically stationary and ergodic.
\end{proof}
%-----------------------------------------------------------------------------------

%-----------------------------------------------------------------------------------
\paragraph{\textit{Build-up Properties of the ARGP Process.}}
The following result shows that the ARGP process replicates data feature \textsf{DF4}:
\begin{proposition}\label{increaseprop}
The function $f: D(\xi,\sigma)\to\R$ is differentiable.
For $\beta>0$ it is strictly increasing, and for $\beta<1$ it is also strictly larger than the identity except for the left endpoint $0$.
In particular, we have
\[
P(X_t>f(X_{t-1}))=0
\]
and
\begin{equation}\label{determIncrease}
P(X_t>X_{t-1},\, X_t\not=f(X_{t-1}))=0.
\end{equation}
\end{proposition}
As a consequence, $\{X_t\}$ increases deterministically along iterated evaluations of $f$.

\begin{proof}
Since $G=G_{\xi,\sigma}$ and $G^{-1}$ are strictly increasing and differentiable $D(\xi,\sigma)$ and $[0,1]$, respectively, we can focus on $f^\ast$.
It is clear that $f^\ast$ is differentiable on $[0,1]$; to see that $f^\ast$ is strictly increasing, note that $(f^\ast)'(u)=\beta/(\beta+(1-\beta)u)^2>0$.
Similarly, for $u\in(0,1)$ and $\beta<1$, we have $0<\beta+(1-\beta)u<1$ so $f^\ast(u)=u/(\beta+(1-\beta)u)>u$.
Since $G^{-1}(1)=\infty$, $f$ is strictly larger than the identity.
By definition, $X_t$ must either coincide with $f(X_{t-1})$ or with $\epsilon_t$; in the latter case, however, due to the presence of the min-operator, $\epsilon_t$ must be smaller than $f(X_{t-1})$.
\end{proof}

\begin{figure}
\includegraphics[angle=-90,width=12cm]{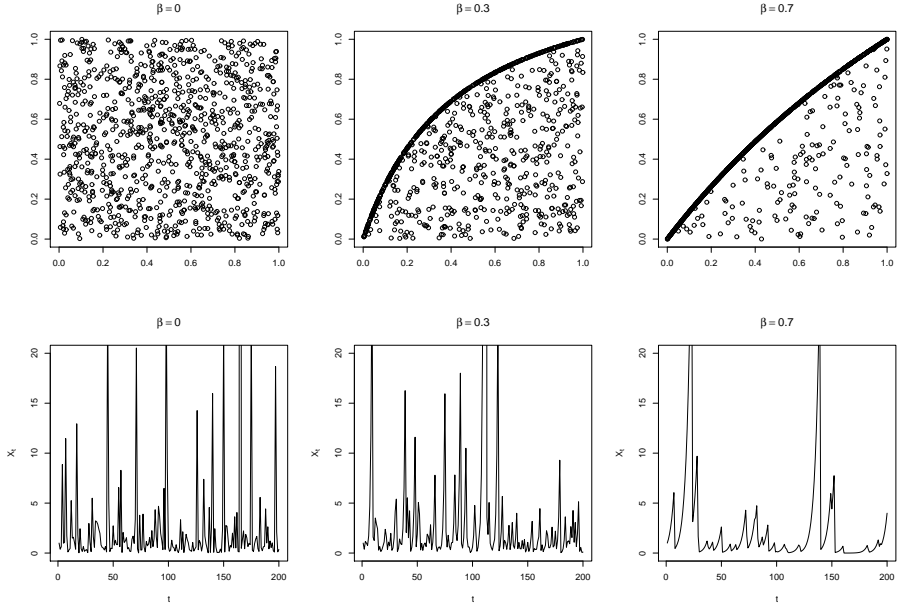}
\caption{Plots of the ARGP process $\{X_t\}$ in terms of lagged, simulated PP plots (i.e.\ simulated values of $G(X_t)$ vs.\ $G(X_{t-1})$), and simulated paths for different values of $\beta$.} \label{Fig2}
\end{figure}

To illustrate property~\eqref{determIncrease}, Figure~\ref{Fig2} displays the probability-integral-transformed subsequent values $G(X_{t-1})$ vs.\ $G(X_t)$
in a lagged PP-plot.
Note that this can be regarded as a visualization of the underlying bivariate copula.
Stationarity and ergodicity warrant that plotting subsequent values is sufficient to represent the (time-invariant) behavior of the two-dimensional marginals.
By property~\eqref{determIncrease}, the area above the curve
$x\mapsto f^\ast(x)$ remains empty.
Note further that this curve, albeit of positive codimension, carries mass.
This is reflected by the fact that there are points lying on the curve.
%---------------------------------------------------------------------------------------------

%---------------------------------------------------------------------------------------------
\paragraph{\textit{Interpretation in terms of Shot-Noise Processes.}}
One of the inherent features of our definition of the ARGP process is its deterministic increase along runs of $U_s=1$, $s=t,t+1,\ldots$
In their spikyness, these stretches of paths resemble shot noise processes, compare \citet{SchmidtStute2007}, i.e.\ a superposition of subsequent shocks whose effects decay through time.

Typically, a shot noise process $\{Y_t\}_{t\in[0,\infty)}$ is defined in a continuous-time setting as a superposition $Y_t=\sum_{i=1}^{N_t} J_i(t-\tau_i)$ of shots $J_i$ that occur at the event times $\tau_i$.
These event times are generated as the jump times of a Cox process $\{N_t\}_{t\in[0,\infty)}$ with intensity $\Lambda_t=\int_0^t \lambda_s\,ds$, so $N_t$ can be represented in terms of a standard Poisson process via a time change, $N_t=\tilde N_{\Lambda_t}$.
In \citet{SchmidtStute2007}, the jump sizes $J_i$ are given as (strong) solutions of SDEs
\[
J_i(t)=J_i(0) + \int_0^t a(s,J_i(s),\eta_i)\,ds + \int_0^t b(s,J_i(s),\eta_i)\,dW_s
\]
where $\eta_i$ is a sequence of i.i.d.\ innovations.

As a discrete-time analog, we replace the Cox process $\{N_t\}$ with a sequence of geometrically distributed holding times.
More precisely, we consider the process $\{V_s\}_{s\in\N_0}$ given by $V_s=\Jc(U_s=0, \epsilon_s<f(X_{s-1}))$ so that on $\{V_s=1\}$ a decrease in the process at time $s$ is triggered.
Second, we replace the jump sizes by a sequence of GPD marginals.
Third, to capture the ``reinitialization of the process'' after an event $\{V_s=1\}$, we generalize the definition of jump sizes so that $J_i$ may depend on past $J_j$.

Then we can embed the ARGP process $\{X_t\}_{t\in\N_0}$ into the framework of \citet{SchmidtStute2007} if we set $\eta_i=(U_i,\epsilon_i)$ and $J_i=J_i(t,x)$, where $J_i(k,x)=f^k(\epsilon_i)-f^k(x)$ and $f^0$ is the identity.
We denote by $\tau_i=\min\{k\colon \sum_{\ell=1}^k V_\ell=i\}$ the time of the $i$-th occurrence of an event $\{V_s=1\}$ and by $N_t=\sum_{\ell=1}^t V_\ell$ the number of events occurred up to time $t$.
Then we have the shot-noise representation
\[
X_t=\sum_{i=1}^{N_t} J_i(t-\tau_i,X_{\tau_i-1}),\quad X_0=0.
\]
Thus the ARGP process not only visually resembles, path by path, a time-reversed shot noise process, but can in fact be interpreted formally as a (generalized) shot noise run backwards in time.
%-----------------------------------------------------------------------------------

%-----------------------------------------------------------------------------------
\paragraph{\textit{The Autocorrelation Structure of ARGP Random Variables.}}
In this paragraph, we consider the moments and correlations of subsequent values of the ARGP process.
As $\sigma$ is a simple scale parameter, $X_t/\sigma \sim G_{\xi,1}$, so we may assume without loss that $\sigma=1$.
Under this assumption, the $r$-th moment of $X_t$ is given by ($r>0$)
\[
m_r=\E X_t^r= \int_0^1 G^{-1}(s)\,ds = \xi^{-r}\int_0^1 (y^{-\xi}-1)^r\,dy
\]
where $G=G_{\xi,1}$ and $\xi=0$ is obtained as a limiting case as usual.
Now $m_r$ is finite as long as $r\xi<1$.
In particular, for $\xi<1$ we obtain $m_1=\E X_t=1/(1-\xi)$, and for $\xi<1/2$ we have $m_2=\E X_t^2=2/((1-2\xi)(1-\xi))$; we assume that $\xi<1/2$ for the subsequent discussion.
We first compute $H_2(x)\defined \E[X_t|X_{t-1}^\ast=x]$.
By \eqref{eq:unifX},
\[
H_2(x)=\frac{\beta}{(1-\beta)x+\beta}\,G^{-1}(f^\ast(x))+ (1-\beta) \int_0^{f^\ast(x)} \hspace{-0.2cm}G^{-1}(u) du
\]
and we obtain
\begin{align*}
\Cov[X_t,X_{t-1}] &= \E X_tH_2(G(X_t)) - m_1^2 = \int_0^1 G^{-1}(s)H_2(s)\,ds - m_1^2\,,\\
\Cor[X_t,X_{t-1}] &= \Cov[X_t,X_{t-1}]/(m_2-m_1^2)\, .
\end{align*}
%-----------------------------------------------------------------------------------

%---------------------------------------------------------------------------------------------
\subsection{The Modified ARGP Process}\label{carmouflage}
%---------------------------------------------------------------------------------------------

%---------------------------------------------------------------------------------------------
As shown in Proposition~\eqref{increaseprop}, increases of the ARGP process occur deterministically along iterated applications of the function $f$ according to
equation~\eqref{determIncrease}, i.e.\ along runs of $\{U_s=1\}\cup\{\epsilon_s> f(X_{s-1})\}$ for $s=t_0,t_0+1,\ldots, t_1$.

If the data suggest that \textsf{DF4} is not required to hold globally, it may be desirable to admit additional regimes where stochastic increases are possible.
This can be achieved naturally within our framework by including an additional switching layer.

Formally, with $G=G_{\xi,\sigma}$ we introduce two further independent processes $\tilde\epsilon_t \iidsim G$ and $W_t\iidsim{\rm Bernoulli}(\gamma)$ and define the {\it modified auto\-regressive generalized Pareto (MARGP) process} $\{\tilde X_t\}_{t\in\N_0}$ by
\begin{align*}
\tilde X_t &\defined (1-W_t)\tilde\epsilon_t + W_t X_t,\\
X_t&\defined U_t f( X_{t-1}) + (1-U_t) \min\{f(X_{t-1}), \epsilon_t\}
\end{align*}
where $\tilde X_0 \sim G$.
If we pass to the transformed values $X^\ast_t=G(X_t)$ and $\tilde X^\ast_t=G(\tilde X_t)$ as in Section~\ref{SectARGPProcess}, we can equivalently represent the MARPG process as
\begin{align*}
\tilde X_t^\ast &= (1-W_t)\tilde\epsilon_t^\ast + W_t X_t^\ast,\\
X_t^\ast &= U_t f^\ast( X^\ast_{t-1}) + (1-U_t) \min\{f^\ast( X^\ast_{t-1}), \epsilon_t^\ast\}
\end{align*}
where $\tilde\epsilon_t^\ast \iidsim {\rm unif}[0,1]$ and $\tilde X_0^\ast \sim {\rm unif}[0,1]$.

It follows as before that $\{X_t^\ast\}$ has uniform marginals, and that $\{X_t\}$ has $G$ marginals.
Considering the lagged PP-plot in Figure~\ref{Fig3}, it becomes apparent that for $\gamma<1$ we have $P(\tilde X_t>f(\tilde X_{t-1}))>0$, i.e.\ there is a non-deterministic increase with positive probability.
Furthermore, this behavior becomes more pronounced for large values of $\gamma$.

\begin{figure}[htb]
\includegraphics[angle=-90,width=12cm]{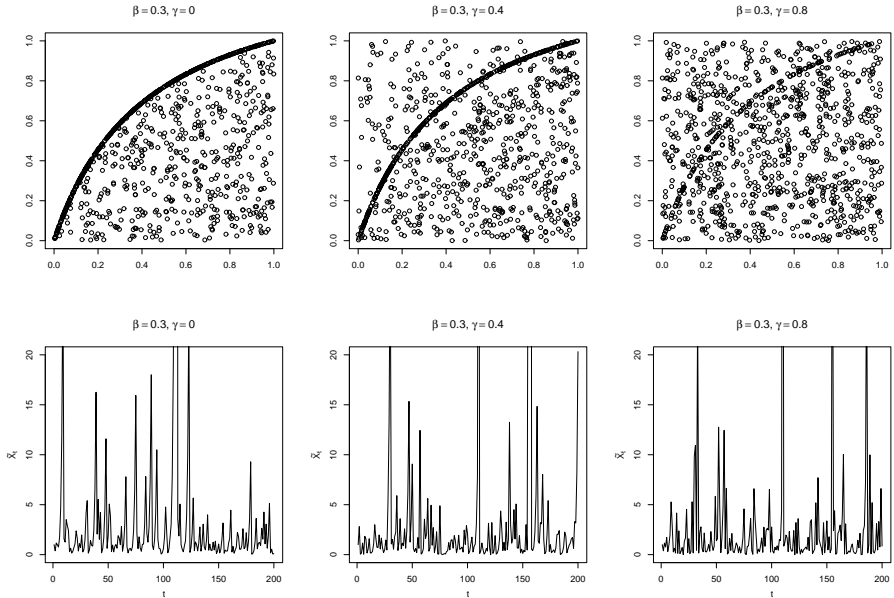}
\caption{Plots of the MARGP process $\{\tilde X_t\}$ in terms of lagged, simulated PP plots and simulated paths for several values of $\gamma$.}\label{Fig3}
\end{figure}
%-----------------------------------------------------------------------------------

%-----------------------------------------------------------------------------------
\subsection{The Truncated ARGP Process} \label{SectTruncARGP}
%-----------------------------------------------------------------------------------

%-----------------------------------------------------------------------------------
\paragraph{\textit{Attaining Zero as a Data Feature}}
As indicated by data feature~\textsf{DF7}, it may be desirable for a process modeling fund redemptions to be able to attain zero.
This reflects trading days without cashflow redemptions for the corresponding share class.
However, neither the ARGP process of Definition~\ref{Def:ARGP} nor its modified version attain zero.

To accommodate this, recall that by the Theorem of Pickands, Balkema and de Haan (see, e.g., \cite{Embrechts}) there exists a function $\sigma(u)$ such that
\[
\lim_{u\uparrow x_F} \sup_{0 < x < x_F-u} \left| F_u(x) - G_{\xi,\sigma(u)}(x) \right| = 0\,,
\]
where $G_{\xi,\sigma(u)}$ is the generalized Pareto distribution as given in Definition~\ref{Def:GPD}, and $x_F$ is a fixed right endpoint in $(-\infty, \infty]$.
This suggests that we can regard the generalized Pareto distribution, and the (M)ARGP process, as models for exceedances over and above a given, high threshold $u$.
%-----------------------------------------------------------------------------------

%-----------------------------------------------------------------------------------
\paragraph{\textit{Truncated ARGP Process}}
Motivated by this, we now introduce a truncated version of the MARGP process that is able to capture this data feature.
Of course, similar considerations also hold for the ARGP process from Section~\ref{carmouflage}.
To avoid repetitions, we focus on the MARGP process in the following; the corresponding considerations for the ARGP process are obtained as the special case $\gamma=P(W_t=1)=1$.

\begin{definition}
For a fixed $u\ge 0$, the {\it truncated autoregressive generalized
Pareto (TARGP) process} $\{V_t\}$ is defined as
\[
V_t  \defined \bigl(\tilde X_t - u\bigr)_+\quad \text{for }t\in\N_0.
\]
\end{definition}

\begin{figure}
\includegraphics[angle=-90,width=12cm]{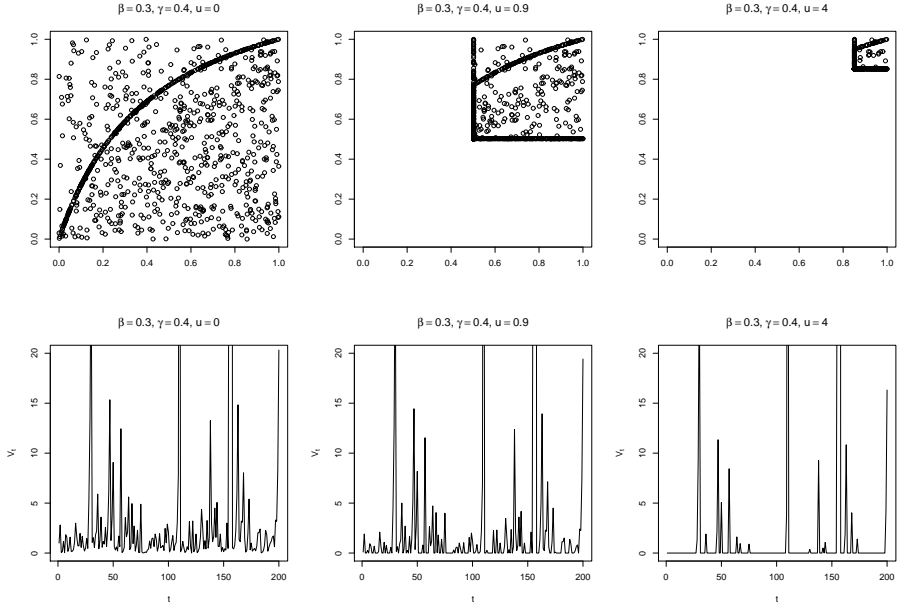}
\caption{Plots of the TARGP process $\{V_t\}$ in terms of lagged, simulated PP plots and as simulated paths for several values of $u$.} \label{Fig4}
\end{figure}

The effect of the censoring is visualized in the lagged PP-plot in Figure~\ref{Fig4}:
With marginal probability $u^\ast$ we have a censoring.
By construction, we do not observe values below $u^\ast$ in $\{V_t^\ast\}$; nevertheless, in the upper panels we display subsequent values of $X_t^\ast$, so at censoring we get placed on the exact probability value of $u^\ast$.
This generates the $L$-shaped corners in the middle and right panels.
This effect, too, becomes more pronounced for larger values $u^\ast$.

\newpage
%---------------------------------------------------------------------------------------------

%-----------------------------------------------------------------------------------
\section{Interarrival Times}\label{InterARR}
%-----------------------------------------------------------------------------------

%-----------------------------------------------------------------------------------
As noted above, for the TARGP process it is possible to derive explicit analytic expressions for the interarrival times of exceedances, i.e.\ of events $\{V_t>0\}$. To this end, defined on the event $\{V_t>0\}$, denote by
\[
L_t \defined \inf\{h>0 \mid V_{t+h}>0\}-1
\]
the number of censorings until the next exceedance; by stationarity of $\{\tilde X_t\}$ (and hence of $\{V_t\}$), the same holds true for $\{L_t\}$.
More specifically we have:
\begin{proposition}[Distribution of Interarrival Times]\label{prop:arrival}
For all $k>0$
\begin{equation}\label{Ltdist}
P(L_t=0)=\pi_1,\quad P(L_t=k)=(1-\pi_1)(1-\pi_0){\pi_0}^{k-1}
\end{equation}
where
\begin{align*}
\pi_1&\defined P(V_t>0 \mid V_{t-1}>0)=1-(1-\beta \gamma) u^\ast\\
\pi_0&\defined P(V_t=0 \mid V_{t-1}=0)=\bar \gamma u^\ast +\gamma \frac{\beta+\bar \beta^2 u^\ast\bar u^\ast}{ \beta +\bar \beta \bar u^\ast}
\end{align*}
and $\bar \gamma=1-\gamma$, $\bar\beta=1-\beta$, $\bar u^\ast = 1-u^\ast$.
In particular,
\[
\E L_t = (1-\pi_1)/(1-\pi_0),\quad \Var L_t = (1-\pi_1)(\pi_0+\pi_1)/(1-\pi_0)^2 \,.
\]
\end{proposition}
Thus $L_t$ is distributed according to a mixture of a Dirac measure at $0$ (with mixture probability $\pi_1$) and a geometric distribution with parameter $1-\pi_0$. In particular, on the set $\{L_t>0\}$, the distribution is memoryless.
\begin{proof}
The Markovian structure of $\tilde X_t$ translates into a Markovian structure for $J_t=\Jc(V_t>0)$.
More specifically, $\{J_t\}$ is a two-state Markov chain with states $0$ and $1$ and transition matrix
\[
\Pi=\left(\begin{array}{cc}
\pi_0 & 1-\pi_0\\
1-\pi_1 & \pi_1
\end{array}\right)
\]
where $\pi_0=P(V_t=0 \mid V_{t-1}=0)$ and $\pi_1=P(V_t>0 \mid V_{t-1}>0)$.
This shows \eqref{Ltdist}.
For the expressions for $\pi_0$ and $\pi_1$, we note that
\begin{align*}
& \{V_t>0,\, V_{t-1}>0\} = \{W_t=0,\, \tilde\epsilon_t^\ast>u^\ast,\, \tilde X_{t-1}^\ast>u^\ast\}\\
&\quad \dot\cup\, \{W_t=1,\, U_t=1,\, f^\ast(\tilde X_{t-1}^\ast)>u^\ast,\, \tilde X_{t-1}^\ast>u^\ast\}\\
&\quad \dot\cup\, \{W_t=1,\, U_t=0,\, \min\{\epsilon^\ast_t,f^\ast(\tilde X_{t-1}^\ast)\}>u^\ast,\, \tilde X_{t-1}^\ast>u^\ast\}\\
=\ &\{W_t=0,\, \tilde\epsilon_t^\ast>u^\ast,\, \tilde X_{t-1}^\ast>u^\ast\}\\
&\quad \dot\cup\, \{W_t=1,\, U_t=1,\, \tilde X_{t-1}^\ast>u^\ast\}\\
&\quad \dot\cup\, \{W_t=1,\, U_t=0,\, \epsilon^\ast_t>u^\ast,\, \tilde X_{t-1}^\ast>u^\ast\}
\end{align*}
and hence
\[
\pi_1=\bar \gamma \bar u^\ast +\gamma \beta + \gamma \bar \beta \bar u^\ast=1-(1-\beta \gamma) u^\ast\,.
\]
As for $\pi_0$, we obtain using Proposition~\ref{increaseprop}
\begin{align*}
& \{V_t=0,\, V_{t-1}=0\} = \{W_t=0,\, \tilde\epsilon_t^\ast \leq u^\ast,\, \tilde X_{t-1}^\ast\leq u^\ast\}\\
&\quad \dot\cup\, \{W_t=1,\, U_t=1,\, f^\ast(\tilde X_{t-1}^\ast) \leq u^\ast,\, \tilde X_{t-1}^\ast\leq u^\ast\}\\
&\quad \dot\cup\, \{W_t=1,\, U_t=0,\, \tilde X_{t-1}^\ast\leq u^\ast\}\setminus \{\epsilon^\ast_t> u^\ast,\, f^\ast(\tilde X_{t-1}^\ast)>u^\ast\}\\
=\ & \{W_t=0,\, \tilde\epsilon_t^\ast \leq u^\ast,\, \tilde X_{t-1}^\ast \leq u^\ast\}\\
&\quad \dot\cup\, \{W_t=1,\, U_t=1,\, \tilde X_{t-1}^\ast\leq (f^\ast)^{-1}(u^\ast)\}\\
&\quad \dot\cup\, \{W_t=1,\, U_t=0,\, \tilde X_{t-1}^\ast\leq u^\ast \}\setminus \{\epsilon^\ast_t> u^\ast,\, \tilde X_{t-1}^\ast>(f^\ast)^{-1}(u^\ast)\}
\end{align*}
and hence, recalling from \eqref{Equafstarinverse} that $(f^\ast)^{-1}(u^\ast) = (\beta u^\ast)/(1-(1-\beta)u^\ast) = (\beta u^\ast)/(\beta+\bar\beta\bar u^\ast)$, after some calculations we get
\begin{align*}
\pi_0 &= \bar \gamma u^\ast +\gamma \beta \frac{(f^\ast)^{-1}(u^\ast)}{u^\ast} + \gamma\bar\beta \frac{u^\ast-(u^\ast-(f^\ast)^{-1}(u^\ast))\bar u^\ast }{u^\ast}\\
%&= \bar \gamma u^\ast +\gamma \beta \frac{\beta}{\beta+\bar\beta\bar u^\ast} + \gamma\bar\beta \left[ 1 - \left(1-\frac{\beta}{\beta+\bar\beta\bar u^\ast}\right)\bar u^\ast \right]\\
%&= \bar \gamma u^\ast +\gamma \beta \frac{\beta}{\beta+\bar\beta\bar u^\ast} + \gamma\bar\beta \left[ u^\ast + \frac{\beta \bar u^\ast}{\beta+\bar\beta\bar u^\ast} \right]\\
%&= \bar \gamma u^\ast + \frac{\gamma}{\beta+\bar\beta\bar u^\ast} \left[ \beta^2 + \bar\beta \left( u^\ast(\beta+\bar\beta\bar u^\ast) + \beta\bar u^\ast\right) \right]\\
%&= \bar \gamma u^\ast + \frac{\gamma}{\beta+\bar\beta\bar u^\ast} \left[ \beta^2 + \bar\beta \left( \beta u^\ast + \bar\beta u^\ast\bar u^\ast + \beta \bar u^\ast \right) \right]\\
%&= \bar \gamma u^\ast + \frac{\gamma}{\beta+\bar\beta\bar u^\ast} \left[ \beta^2 + \bar\beta^2 u^\ast \bar u^\ast + \bar\beta\beta \right]\\
&= \bar \gamma u^\ast + \gamma \frac{\beta + \bar\beta^2 u^\ast \bar u^\ast}{\beta+\bar\beta\bar u^\ast}\,. \qedhere
\end{align*}
\end{proof}

\newpage
%-----------------------------------------------------------------------------------

%-----------------------------------------------------------------------------------
\section{Parameter Estimation} \label{ParamEst}
%-----------------------------------------------------------------------------------

%-----------------------------------------------------------------------------------
This section provides statistical estimators for the underlying parameters of the different variants of the ARGP process.
To this end, we first address smoothness of the model and subsequently provide
parameter estimators for the original ARGP process and its modifications.
%-----------------------------------------------------------------------------------

%-----------------------------------------------------------------------------------
\subsection{Parameter Estimation for the ARGP Process}
%-----------------------------------------------------------------------------------

%-----------------------------------------------------------------------------------
As indicated above, for inference we assume that the shape parameter $\xi$ of the marginals is larger than $-1/2$. By the ergodic theorem (see for instance Theorem~6.21 in \cite{Breiman92}), the smoothness of the GPD cdf in $\sigma$ and $\xi$, and the $\alpha$-mixing property we obtain
\begin{theorem}\label{MLETHM}
For $\xi> -1/2$ the following statements hold:
\begin{enumerate}
\item[{(a)}] The ARGP model is $L_2$-differentiable.
\item[{(b)}] The Maximum Likelihood estimators (MLEs), evaluated on a single path $\{X_t(\omega)\}$ as if the values $X_t$ were realizations of an i.i.d.\ sequence, are strongly consistent and asymptotically normal.
\item[{(c)}] $\beta$ can be estimated by $\hat\beta_n\defined 1-2\hat p_{n} $ where $\hat p_{n}$ is the empirical counterpart of the probability $p\defined P(X_t<X_{t-1})$. $\hat \beta_n$ is again strongly consistent, and, together with the MLE for $\sigma,\xi$, jointly asymptotically normal.
\end{enumerate}
\end{theorem}
\begin{remark}
For $\xi\leq -1/2$ the ARGP model fails to be $L_2$-differentiable, and
observations scattered around the right endpoint of the distribution $-\sigma/\xi$
become overly informative.
Moreover, for $\xi < -1/2$, similarly to the situation when one wants to estimate the endpoint $\theta$ in ${\rm unif}[0,\theta]$, it pays off to only use the maximal observation $M_n=\max_{1 \leq t\leq n} X_t$.
In particular, in this case MLE is generally no longer asymptotically normal and will be consistent at a higher rate than the usual $1/\sqrt{n}$, where $n$ is the sample size.
For more details on these situations, see, e.g., Theorem~3 in \cite{Smith85}. \remaend
\end{remark}
\begin{proof}
\begin{enumerate}
\item[(a)] For $L_2$-differentiability of the ARGP model, let us first deal with the i.i.d.\ situation.
In this case, there are two distinct situations.
For $\xi\ge 0$, the support $D(\xi,\sigma)$ of the distributions does not depend on the parameter, so we may argue as in \citet{Haj:72} (only treating a one dimensional parameter) and \citet[Satz~1.194]{Witting85}.
For $\xi\in(-1/2,0)$, one has to treat the two subdomains $D(\xi,\sigma)\,{\scriptstyle \Delta}\,D(\xi',\sigma)$ and $D(\xi,\sigma)\,\cap\,D(\xi',\sigma)$ separately.
For the latter we may again argue as for $\xi\ge0$, for the former one shows that the contribution to the squared integrated remainder is of order
$\mathop{\rm o}(|\xi-\xi'|^2)$.
Transferring the integration to $[0,1]$ by the quantile transformation, this easily follows.
An alternative proof can be given along the lines of Proposition~3.2 in \cite{BuecherSegers16}.
The translation from the i.i.d.\ case into our $X_t$-setting follows from ergodicity and (asymptotic) stationarity, see Proposition~\ref{alphamix}.
\item[(b)] Part (a) gives the regularity conditions under which (suitable versions of) the MLEs are well defined and strongly consistent.
Asymptotic normality follows from the $\alpha$-mixing property, see Proposition~\ref{alphamix}, by a ramification of the central limit theorem for stationary, $\alpha$-mixing variables; see, e.g., \citet[Theorem~18.5.3]{IbragimovLinnik}.
\item[(c)] By the ergodic theorem, $\hat p_{n}$ converges almost surely to $p$, implying strong consistency.
Now
\[
\{X_t<X_{t-1}\}=\{X_t^\ast<X_{t-1}^\ast\}=\{U_t=0,\, \epsilon^\ast_t<X_{t-1}^\ast\}
\]
so by symmetry of the i.i.d.\ variables $\epsilon^\ast_t$ and $X_{t-1}^\ast$ we obtain
\[
p=P(X_t<X_{t-1}) = (1-\beta)/2.
\]
Hence $\beta=1-2p$, so $\hat\beta$ is strongly consistent for $\beta$.
As in (b), $\alpha$-mixing (together with the Cram\'er-Wold-device) implies joint asymptotic normality. \qedhere
\end{enumerate}
\end{proof}
%-----------------------------------------------------------------------------------

%-----------------------------------------------------------------------------------
\subsection{Parameter Estimation for the MARGP Process}
%-----------------------------------------------------------------------------------

%-----------------------------------------------------------------------------------
In the MARGP model, the estimation of $\gamma$ and $\beta$ becomes somewhat more involved.
We start with the simplifying assumption that $f$ is known (which in fact assumes
knowledge of $\xi, \sigma, \beta$), so we only require one additional estimation equation to estimate $\gamma$.

\begin{proposition}\label{prop34}
Let $\tilde p_{n}$, $\tilde q_n$ be the empirical counterparts of the probabilities
\[
\tilde p \defined P(\tilde X_t<\tilde X_{t-1}),\quad \tilde q\defined P(\tilde X_t>\tilde X_{t-1},\, \tilde X_t\not=f(\tilde X_{t-1})).
\]
Then $\tilde\gamma_n\defined 1-2\tilde q_n$ and $\tilde\beta_n\defined (1-2\tilde p_{n})/(1-2\tilde q_{n})$ are strongly consistent, and, together with the MLE for $\sigma,\xi$, jointly asymptotically normal.
\end{proposition}
\begin{proof}
Note that
\begin{align*}
\{\tilde X_t<\tilde X_{t-1}\} &= \{\tilde X_t^\ast<\tilde X_{t-1}^\ast\}\\
&= \{W_t=0,\,\tilde\epsilon_t^\ast <\tilde X_{t-1}^\ast\}\, \dot\cup\,
\{W_t=1,\, U_t=0,\, \epsilon_t^\ast <\tilde X_{t-1}^\ast\}
\end{align*}
so $\tilde p=(1-\gamma)/2+ \gamma (1-\beta)/2=1/2-\beta\gamma/2$.
For the second event, note that $\{\tilde X_t>\tilde X_{t-1},\,\tilde X_t \not=f(\tilde X_{t-1})\}$ may be equivalently written as
\[
\{W_t=0,\,\tilde\epsilon_t>\tilde X_{t-1},\,\tilde\epsilon_t\not=f(\tilde X_{t-1})\}
\]
and hence by symmetry carries probability $\tilde q=(1-\gamma)/2$.
Thus, if $\tilde p$, $\tilde q$ are the empirical frequencies $\tilde p_{n}$, $\tilde q_{n}$ of the respective events, the estimators $\tilde\gamma_n$ and $\tilde\beta_n$ are strongly consistent by the ergodic theorem.
Joint asymptotic normality follows as in Theorem~\ref{MLETHM}.
\end{proof}

While $\sigma$ and $\xi$ can indeed be estimated separately, to be able to compute
$\tilde q_n$, we need a preliminary estimator for $\beta$.
To make this dependence
explicit, in this paragraph, we write $f_\beta$ for $f$.
The curve $\{(x,f_\beta(x))\}_{x\in[0,1]}\subset [0,1]^2$ is a two-dimensional Lebesgue null set; hence, if for some $\beta$, $\{(x,f_\beta(x))\}$ carries mass, it must be the true $\beta$.
This can be used for a preliminary estimator for $\beta$ as follows:
To some given small $\epsilon>0$ and some grid $\{\beta_i,\;1\leq i \leq M\}\subset[0,1]$, we compute $N_i \defined \sum_{j=2}^n\Jc(|\tilde X_{j}-f_{\beta_i}(\tilde X_{j-1})|<\epsilon)$ and use $\beta_n^{(0)}\defined\beta_{i_0}$ as an initial estimator where $N_{i_0}=\max_{1\leq i\leq M}N_i$.
With this $\beta_n^{(0)}$, and the separately computed MLE for $\sigma,\xi$, we compute $\tilde q_n$; in a second step, we determine the relevant estimators for $\beta$ and $\gamma$ using Proposition~\ref{prop34}.
%-----------------------------------------------------------------------------------

%-----------------------------------------------------------------------------------
\subsection{Parameter Estimation for the TARGP Process}
%-----------------------------------------------------------------------------------

%-----------------------------------------------------------------------------------
As in the previous subsections, we simplify the following discussion without loss by passing to $V^\ast=G(V)$.
We denote the censoring probability by
\[
u^\ast  \defined G(u)\,.
\]
The following result is a straightforward generalization of Proposition~\ref{prop34} to the TARGP process:
\begin{proposition}
Let $\tilde p^{u}_{n}$, $\tilde q^{u}_n$ be the empirical counterparts of the probabilities
\[
\tilde p^{u}\defined P(V_t<V_{t-1}),\quad \tilde q^{u}\defined P(V_t>V_{t-1},\,V_t\not=f(V_{t-1}))
\]
and set $\bar u^\ast_2 \defined 1- (u^\ast)^2$.
Then the estimators
\begin{align*}
\tilde\gamma^{u}_n &\defined 1-2\tilde q^{u}_n/\bar u^\ast_2\\
\tilde\beta^{u}_n &\defined (\bar u^\ast_2-2\tilde p^{u}_{n})/(\bar u^\ast_2-2\tilde q^{u}_{n})
\end{align*}
are strongly consistent, and, together with the MLE for $\sigma,\xi$, jointly asymptotically normal.
\end{proposition}
Hence we see that, formally, using censoring amounts to replacing $1$ by $\bar u^\ast_2$ in Proposition~\ref{prop34}.
\begin{proof}
By monotonicity we have $V_t^\ast=(\tilde X_t^\ast -u^\ast)_+$ and hence
\begin{align*}
\{V_t< V_{t-1}\} = \{V_t^\ast<V_{t-1}^\ast\} &= \{W_t=0,\, \tilde\epsilon_t^\ast <V_{t-1}^\ast,\, u^\ast<V_{t-1}^\ast\}\\
&\quad \dot\cup\, \{W_t=1,\, U_t=0,\, \epsilon_t^\ast <V_{t-1}^\ast,\, u^\ast <V_{t-1}^\ast\}.
\end{align*}
Thus we obtain
\begin{align*}
&P(W_t=0,\, \tilde\epsilon_t^\ast <V_{t-1}^\ast,\, u^\ast <V_{t-1}^\ast)
= (1-\gamma) \int_{u^\ast}^1 P(\tilde\epsilon_t^\ast< v)\,dv\\
=\ & (1-\gamma) \int_{u^\ast}^1 v\,dv=(1-\gamma)(1-(u^\ast)^2)/2=(1-\gamma)\bar u^\ast_2/2\, .
\end{align*}
A similar argument applies to $\{W_t=1,\, U_t=0,\, \epsilon_t^\ast <V_{t-1}^\ast,\, u^\ast <V_{t-1}^\ast\}$, so we get
\[
\tilde p^{u}=(1-\gamma)\bar u^\ast_2/2+ \gamma (1-\beta)\bar u^\ast_2/2=(1-\beta\gamma)\bar u^\ast_2/2.
\]
For the second event, we represent $\{V_t>V_{t-1}, V_t\not=f(V_{t-1})\}$ as
\[
\{W_t=0,\,\tilde\epsilon_t>V_{t-1},\,\tilde\epsilon_t\not=f(V_{t-1}),\,u^\ast <V_{t-1}^\ast\}
\]
so it carries probability $\tilde q^{u}=(1-\gamma)\bar u^\ast_2/2$.
Hence if we determine $\tilde p^{u}$, $\tilde q^{u}$ as empirical frequencies $\tilde p^{u}_{n}$, $\tilde q^{u}_{n}$ of these events, the estimators $\tilde\gamma^{u}_n$ and $\tilde\beta^{u}_n$ are strongly consistent.
Joint asymptotic normality follows as in Theorem~\ref{MLETHM}.
\end{proof}
To be able to evaluate $f=f_{\beta,\xi,\sigma}$, as in Section~\ref{carmouflage}, we need a preliminary estimator for $\beta$, which can be computed similarly as $\beta_n^{(0)}$.
%-----------------------------------------------------------------------------------

%-----------------------------------------------------------------------------------
\subsection{Threshold Estimation}
%-----------------------------------------------------------------------------------

%-----------------------------------------------------------------------------------
Estimating the threshold $u$ from data is a delicate issue---mainly because the Pickands-Balkema-de-Haan theorem only justifies the GPD for the conditional distribution exceeding a threshold converging to infinity, so a bias can arise when the GPD is used at finite thresholds; in particular, usual $\sqrt{n}$-type asymptotics are no longer applicable.
For strategies to overcome this issue, see, e.g., \cite{CsoergoeEtAl85, Drees96, McNeil96, Embrechts, DreesKaufmann98, ChoulakianStephens01, BaudEtAl02, BeirlantEtAl2, Thompson2009,Deidda2010}, and the survey \citet{DuttaPerry06};
in the context of liquidity risk for cashflow redemptions, this is also detailed in \cite{DesmettreDeege2016}.

For this paper we assume knowledge of $u$ and, by usual location equivariance arguments, we reduce the discussion of the marginal Pareto distributions to the situation of a $0$ threshold by passing over to the exceedances $X-u$; of course, the scale $\sigma(u)$ in the censored model is transformed according to
\[
\sigma(u)\defined \sigma + \xi u\,.
\]
%-----------------------------------------------------------------------------------

%-----------------------------------------------------------------------------------
\section{Application to Fund Redemption Data} \label{SectResults}
%-----------------------------------------------------------------------------------

%-----------------------------------------------------------------------------------
The data we use to illustrate our TARGP process are from IPConcept (Luxemburg) S.A., a Luxembourg-based management company.
It includes absolute cashflows $AF^\redeem_t = SR_t \cdot RP_{t-1}$ paid to investors for shares redeemed, where $RP$ denotes the redemption price of a single share and $SR$ denotes the number of shares redeemed by the fund.
The data comprises historical data from August 22, 2000 to February 22, 2016 on a daily basis ($3888$ data points in total), for open ended mutual funds that are governed by the European UCITS and AIFM directives.
%-----------------------------------------------------------------------------------

%-----------------------------------------------------------------------------------
\paragraph{\textit{Simulated Realizations.}} Using the estimators $\tilde\beta^u_n$ and $\tilde\gamma^u_n$ as well as the calibrated (estimated) values of $u$, $\xi$ and $\sigma$ for a particular share class, we simulate a path $V_t^\simu$ of the TARGP process and compare it with the time series of the market data $V_t^\data$ from the introductory example of Figure~\ref{fig:Example}; the result of this is displayed in Figure~\ref{fig:SimulatedFund}, where we  have used the following estimated parameters:
\[
u =  2168\,,\quad \xi =  0.5538\,,\quad \sigma = 11488\,,\quad \tilde\beta^u_n = 0.8619 \,,\quad \tilde\gamma^u_n = 0.5778 \,.
\]
We observe that the TARGP process is able to produce redemption time series whose order of magnitude and qualitative behavior are virtually indistinguishable from market data.
In particular, TARGP processes are capable of capturing the two different regimes of fund redemptions as described by \textsf{DF3} in Section~\ref{DataFeatures}.

\begin{figure}[h]
\begin{center}
\mbox{\includegraphics[width=\textwidth]{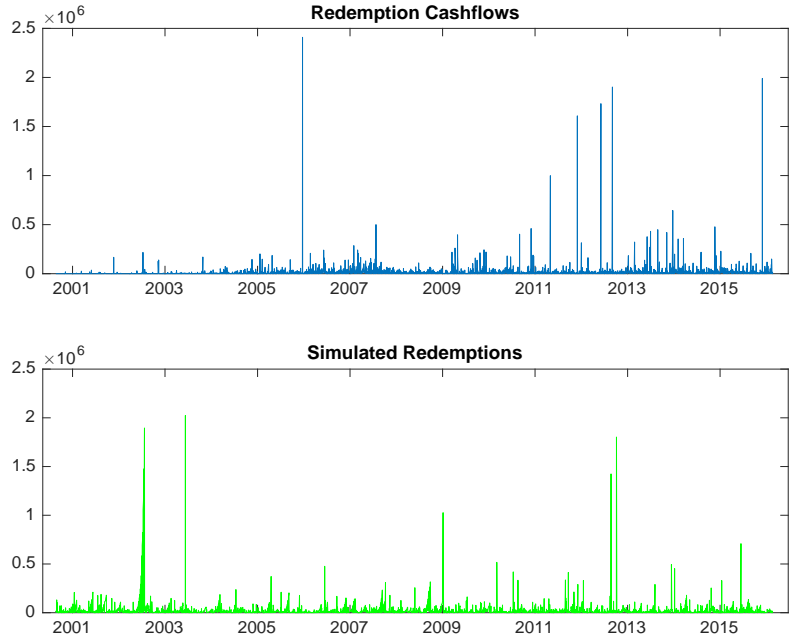}}
\end{center}
\caption{Redemption cashflow paths (upper panel) and simulated realization of the TARGP process (lower panel).}\label{fig:SimulatedFund}
\end{figure}

This is further substantiated by the box plots displayed in Figure~\ref{fig:Box}, as we observe excellent consistency over the entire range.

\begin{figure}[h]
\begin{center}
\mbox{\includegraphics[width=\textwidth,height=7cm]{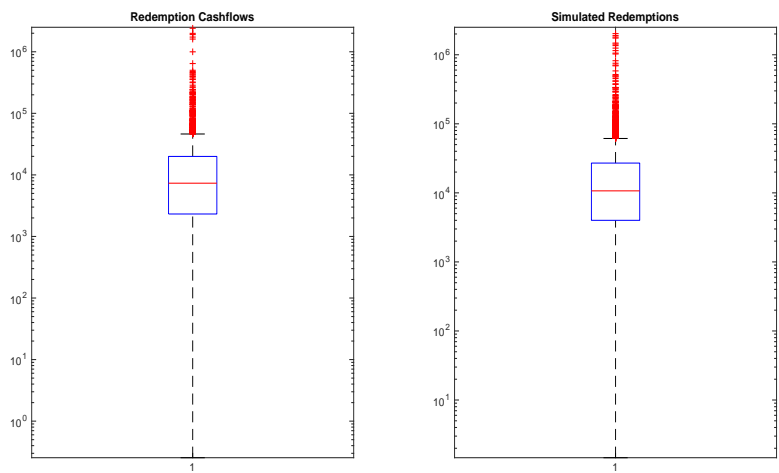}}
\end{center}
\caption{Box plot of redemption cashflows (left) and box plot of simulated TARGP process (right) on a logarithmic scale.}\label{fig:Box}
\end{figure}
%-----------------------------------------------------------------------------------

%-----------------------------------------------------------------------------------
\paragraph{\textit{Interarrival Times.}} To underline the qualitative insights from the visual fit, we also examine the interarrival times of the TARGP process.
For that purpose, we determine the interarrival times of the real data,
\[
L_t^\data = \inf\{h>0 \mid V^\data_{t+h}>0\}-1
\]
and compare them with the interarrival times
\[
L_t^\simu = \inf\{h>0 \mid V^\simu_{t+h}>0\}-1
\]
of a simluated realization of the TARGP process with the same length as the market data time series.

\begin{figure}[h]
\begin{center}
\mbox{\includegraphics[width=\textwidth,height=7cm]{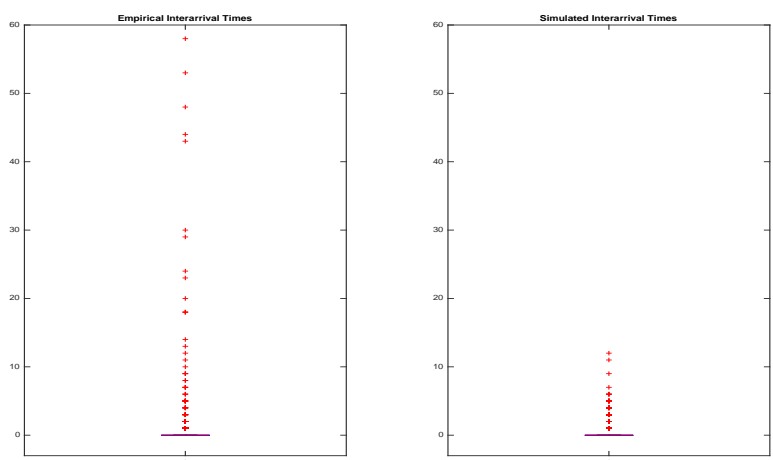}}
\end{center}
\caption{Box plots of interarrival times $L_t^\data$ (left) and $L_t^\simu$ (right).}\label{fig:wt_box}
\end{figure}

Figure~\ref{fig:wt_box} shows the results.
The fits of interarrival times are not fully convincing, at least if one aims for a good fit over the whole range of interarrival times.
The fact that our model---as do most models derived from extreme value asymptotics---generates (almost) memoryless distributions for interarrival times (see Proposition~\ref{prop:arrival}) is not easily reconciled with the interarrivals seen in the data.
An acceptable fit is achieved up to approximately the third quartile, whereas in the tails but in the tails one observes much longer interarrival times in reality than produced with our model.
From a risk modeling perspective, however, it is much more important to adequately capture the behavior of higher frequencies; in this dimension, the TARGP model performs very well.

Upon a closer inspection of the redemption time series (compare Figure~\ref{fig:Example}), we can identify a transient initial phase of the fund, where investors start buying shares in the fund and redemptions are scarce.
This phenomenon is quite well-known in professional fund management, and thus risk management typically and justifiably focuses on the post-burn-in phase.
Hence, the actual goal in the modeling of interarrival times is to fit interarrivals after the burn-in phase.
To examine this and to identify the initial phase, we also calculate the interarrival times $L_t^\data$ and $L_t^\simu$ based on time series starting after certain offsets.
Thus our calulcations are based on the series $V_t^\data$ and $V_t^\simu$ for $t\in [1,N]\,,\,[252,N]\,,\,[504,N]\,,\,[756,N]\,,\,[1008,N]\,,\,[1260,N]$, where $N$ denotes the total length of the original time series.
Here the offsets of $252,504,756,1008,1260$ trading days represent a time scale from $1$ to $5$ years.

\begin{figure}[h]
\begin{center}
\mbox{\includegraphics[width=\textwidth,height=14cm]{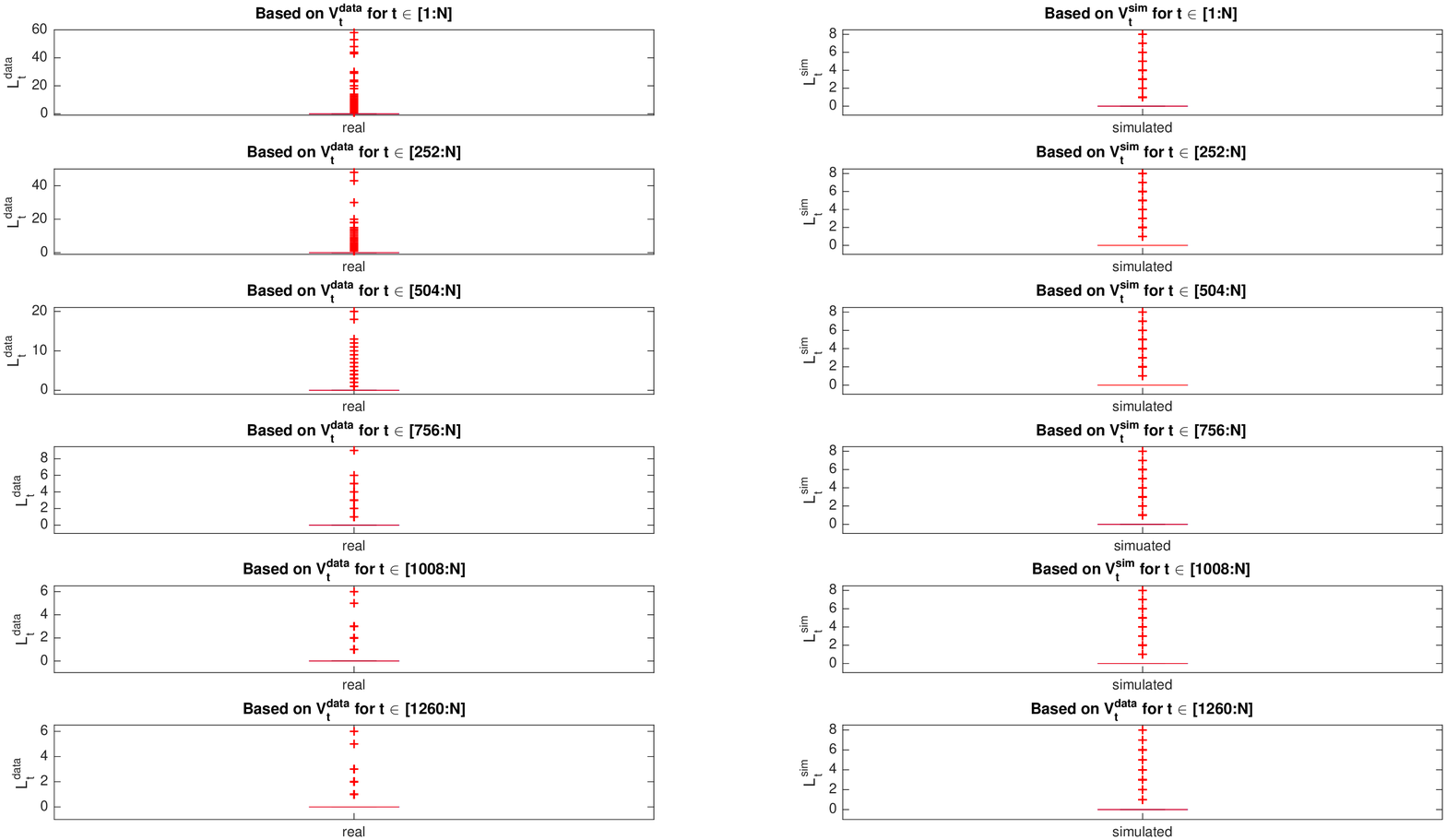}}
\end{center}
\caption{Box plots of interarrival times $L_t^\data$ (left) and $L_t^\simu$ (right) for different lengths of the underlying time series.}\label{fig:wt_compare}
\end{figure}

The results are illustrated in Figure~\ref{fig:wt_compare}.
First, they indicate that indeed the statistics stabilize beyond a certain offset ($4$ years in the displayed example).
Second, we see that the waiting times of the simulated process are quite stable over time.
When the burn-in phase is cut off, the TARGP model achieves a remarkably good fit over the entire range of interarrival times.
Thus, with the exception of funds that are still in their non-stationary initial phase, the TARGP model appears to be a suitable model for interarrival times in our application domain---at least after a burn-in phase, and there is some evidence that the high frequencies of long waiting times seen in the data in Figure\ref{fig:wt_box} could be pre-asymptotic behavior.
%-----------------------------------------------------------------------------------

%-----------------------------------------------------------------------------------
\section{Conclusion}\label{SectConclusion}
This article has proposed a flexible class of time series models that can be used to model dynamic phenomena that feature extreme-value distributions.
We have also provided estimators for the underlying parameters, and we have illustrated our model in the context of fund liquidity management.

The framework developed in this article may be applied as well to daily river discharge data or to portfolios of insurance contracts, in order to get a probabilistic view about the occurance of floodings or the return period of large claims, using in particular the corresponding interarrival times.
%-----------------------------------------------------------------------------------

%-----------------------------------------------------------------------------------
\subsection*{Acknowledgments}
We wish to thank Matthias Deege, Tom Hurd, Gerald Kroisandt, Alfred M\"uller, Bernhard Spangl, Anatoliy Swishchuk, Rudi Zagst, as well as the participants of the DGVFM Workshop "Science meets Practice" 2014 in Kaiserslautern, the Symposium Extreme Events 2014 in Hanover, the CFE 2014 Conference in Pisa, the Challenges in Derivative Markets Conference 2015 in Munich and the 9th Bachelier World Congress in New York for useful comments and discussions.
S.~Desmettre and P.~Ruckdeschel gratefully acknowledge financial support by the Volkswagen Foundation for the project ``Robust Risk Estimation'' and by the DFG within the research grant RU 893/4-1.
S.~Desmettre also thanks the DFG for funding within the Research Training Group 1932 ``Stochastic Models for Innovations in the Engineering Sciences.''
All authors thank IPConcept (Luxemburg) S.A.\ for providing the data sets used in this article.
%-----------------------------------------------------------------------------------

%-----------------------------------------------------------------------------------

%-----------------------------------------------------------------------------------

\end{document}